\documentclass[10pt,draftcls,,onecolumn]{IEEEtran}

\usepackage{graphicx}
\usepackage{amsmath,epsfig,amssymb, amstext, verbatim,amsopn,cite,subfigure,multirow,multicol,lipsum}
\usepackage{balance}
\usepackage[all]{xy}
\usepackage{url}
\usepackage{amsfonts}
\usepackage{epsfig}
\usepackage{epstopdf}
\usepackage{slashbox}
\usepackage{setspace}
\usepackage{stmaryrd}
\usepackage{psfrag}	
\allowdisplaybreaks

\newcommand{\Equal}{\hspace{-0.7mm}=\hspace{-0.7mm}}
\newcommand{\Add}{\hspace{-0.7mm}+\hspace{-0.7mm}}
\newcommand{\Minus}{\hspace{-0.7mm}-\hspace{-0.7mm}}

\newtheorem{theorem}{Theorem}

\newenvironment{proof}[1][Proof]{\begin{trivlist}
\item[\hskip \labelsep {\bfseries #1}]}{\end{trivlist}}

\newenvironment{remark1}[1][Remark 1:]{\begin{trivlist}
\item[\hskip \labelsep {\bfseries #1}]}{\end{trivlist}}
\newenvironment{remark2}[1][Remark 2:]{\begin{trivlist}
\item[\hskip \labelsep {\bfseries #1}]}{\end{trivlist}}

\newcommand{\qed}{\nobreak \ifvmode \relax \else
      \ifdim\lastskip<1.5em \hskip-\lastskip
      \hskip1.5em plus0em minus0.5em \fi \nobreak
      \vrule height0.75em width0.5em depth0.25em\fi}

\ifCLASSINFOpdf
\else
\fi

\begin{document}

\title{Adaptive Mode Selection and Power Allocation in Bidirectional Buffer-aided Relay Networks}

\author{Vahid Jamali$^\dag$, Nikola Zlatanov$^\ddag$, Aissa Ikhlef$^\ddag$, and Robert Schober$^\dag$ \\
\IEEEauthorblockA{$^\dag$ Friedrich-Alexander University (FAU), Erlangen, Germany \\
 $^\ddag$ University of British Columbia (UBC), Vancouver, Canada}
}

\maketitle

\begin{abstract}
In  this  paper,  we  consider  the  problem  of  sum
rate maximization in a bidirectional relay network with fading.
Hereby,  user  1  and  user  2  communicate  with  each  other  only
through a relay, i.e., a direct link between user 1 and user 2 is
not present. In this network, there exist six possible transmission
modes: four point-to-point modes (user 1-to-relay, user 2-to-relay,
relay-to-user  1,  relay-to-user  2),  a  multiple  access  mode  (both
users  to  the  relay),  and  a  broadcast  mode  (the  relay  to  both
users). Most existing protocols assume a fixed schedule of using a
subset of the aforementioned transmission modes, as a result, the
sum rate is limited by the capacity of the weakest link associated
with the relay in each time slot. Motivated by this limitation, we
develop a protocol which is not restricted to adhere to a predefined schedule for
using the transmission modes. Therefore, all transmission modes
of the bidirectional relay network can be used adaptively based
on  the  instantaneous  channel  state  information  (CSI)  of  the
involved links. To this end, the relay has to be equipped with two buffers for the storage of the information received from users 1 and 2, respectively. For the considered network, given a total average
power budget for all nodes, we jointly optimize the transmission
mode selection and power allocation based on the instantaneous
CSI  in  each  time  slot  for  sum  rate  maximization.  Simulation
results  show  that  the  proposed  protocol  outperforms  existing
protocols  for  all  signal-to-noise  ratios  (SNRs).  Specifically,  we
obtain a considerable gain at low SNRs due to the adaptive power
allocation and at high SNRs due to the adaptive mode selection.\end{abstract}
%

\section{Introduction} \label{Sec I (Intro)}
In a bidirectional relay network, two users exchange information via a relay node \cite{Tarokh}. Several protocols have been proposed
for such a network under the practical half-duplex constraint,
i.e., a node cannot transmit and receive at the same time and
in the same frequency band. The simplest protocol is the traditional two-way relaying protocol in which the transmission
is accomplished in four successive point-to-point phases: user
1-to-relay,  relay-to-user  2,  user  2-to-relay,  and  relay-to-user
1.  In  contrast,  the  time  division  broadcast  (TDBC)  protocol
exploits the broadcast capability of the wireless medium and
combines the relay-to-user 1 and relay-to-user 2 phases into
one phase, the broadcast phase \cite{TDBC}. Thereby, the relay broadcasts a superimposed codeword, carrying information for both
user 1 and user 2, such that each user is able to recover its
intended information by self-interference cancellation. Another
existing  protocol  is  the  multiple  access  broadcast  (MABC)
protocol in which the user 1-to-relay and user 2-to-relay phases
are also combined into one phase, the multiple-access phase
\cite{MABC}.  In  the  multiple-access  phase,  both  user  1  and  user  2
simultaneously transmit to the relay which is able to decode
both messages. Generally, for the bidirectional relay network
without a direct link between user 1 and user 2, six transmission  modes are possible:  four  point-to-point  modes
(user 1-to-relay, user 2-to-relay, relay-to-user 1, relay-to-user
2),  a  multiple  access  mode  (both  users  to  the  relay),  and  a broadcast mode (the relay to both users), where the capacity
region of each transmission mode is known \cite{BocheIT}, \cite{Cover}. Using this knowledge,  a significant research effort has been dedicated to obtaining the achievable rate region  of the bidirectional relay network
 \nocite{Tarokh,BocheIT,BochePIMRC,PopovskiICC,PopovskiLetter} \cite{Tarokh}-\cite{ PopovskiLetter}.  Specifically, the achievable rates of  most existing protocols for two-hop relay transmission are
limited by the instantaneous capacity of the weakest link associated with the relay. The reason for this is the fixed schedule of using the transmission modes which is adopted in all existing protocols, and does not exploit the instantaneous channel state information (CSI) of the involved links. For one-way relaying, an adaptive link selection protocol was proposed
in \cite{NikolaJSAC} where based on the instantaneous CSI, in each time slot,
either the source-relay or relay-destination links are selected
for transmission. To this end, the relay has to have a buffer for
data storage. This strategy was shown to achieve the capacity
of the one-way relay channel with fading \cite{NikolaTIT}.

Moreover, in fading AWGN channels, power control
is  necessary  for  rate  maximization.  The  highest  degree  of
freedom  that  is  offered  by  power  control  is  obtained  for  a
joint average power constraint for all nodes. Any other power constraint
with the same total power budget is more restrictive than the joint
power constraint and results in a lower sum rate. Therefore,
motivated by the protocols in \cite{NikolaJSAC} and \cite{NikolaTIT}, our goal is to utilize all available degrees of freedom of the three-node half-duplex
bidirectional relay network with fading, via an adaptive mode
selection  and  power  allocation  policy.  In  particular,  given  a
joint power budget for all nodes, we find a policy which in
each time slot selects the optimal transmission mode from the
six  possible  modes  and  allocates  the  optimal  powers  to  the
nodes  transmitting  in  the  selected  mode,  such  that  the  sum
rate is maximized.

Adaptive mode selection for bidirectional relaying was also considered in \cite{PopovskiLetter} and \cite{EUSIPCO}. However, the selection policy in \cite{PopovskiLetter} does not use all possible modes, i.e., it only selects from two point-to-point modes and the
broadcast mode, and assumes that the transmit powers of all
three  nodes  are  fixed  and  identical.  Although  the  selection
policy  in  \cite{EUSIPCO}  considers  all   possible  transmission  modes  for
adaptive mode selection, the transmit powers of the nodes are assumed to
be fixed, i.e., power allocation is not possible. Interestingly, mode selection and power allocation are mutually coupled and the modes selected with the protocol in \cite{EUSIPCO} for a given channel are different from the modes selected with the proposed protocol. Power allocation can considerably improve the sum rate by optimally allocating the powers to the nodes based on the instantaneous CSI especially when the total power budget in the network is low.  Moreover, the  proposed  protocol
achieves the maximum sum rate in the considered bidirectional
network. Hence, the sum rate achieved with
the proposed protocol can be used as a reference for other low complexity suboptimal protocols. Simulation
results confirm that the proposed protocol outperforms existing
protocols.

Finally, we note that the advantages of buffering come at
the expense of an increased end-to-end delay. However, with
some modifications to the optimal protocol, the average delay
can  be  bounded,  as  shown  in  \cite{NikolaJSAC},  which  causes only  a  small loss in the achieved rate. The delay analysis of the proposed
protocol is beyond the scope of the current work and is left
for future research.

\section{System Model}\label{SysMod}
\begin{figure}
\centering
\psfrag{U1}[c][c][0.75]{$\text{User 1}$}
\psfrag{U2}[c][c][0.75]{$\text{User 2}$}
\psfrag{R}[c][c][0.75]{$\text{Relay}$}
\psfrag{h1}[c][c][0.75]{$h_1(i)$}
\psfrag{h2}[c][c][0.75]{$h_2(i)$}
\psfrag{P1}[c][c][0.75]{$P_1(i)$}
\psfrag{P2}[c][c][0.75]{$P_2(i)$}
\psfrag{Pr}[c][c][0.75]{$P_r(i)$}
\psfrag{B1}[c][c][0.75]{$B_1$}
\psfrag{B2}[c][c][0.75]{$B_2$}
\includegraphics[width=3.5 in]{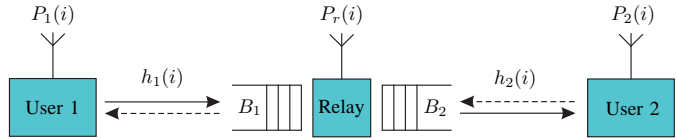}
\caption{Three-node bidirectional relay network consisting of two users and a relay.}
\label{FigSysMod}
\end{figure}
\begin{figure}
\centering
\psfrag{U1}[c][c][0.5]{$\text{User 1}$}
\psfrag{U2}[c][c][0.5]{$\text{User 2}$}
\psfrag{R}[c][c][0.5]{$\text{Relay}$}
\psfrag{M1}[c][c][0.75]{$\mathcal{M}_1$}
\psfrag{M2}[c][c][0.75]{$\mathcal{M}_2$}
\psfrag{M3}[c][c][0.75]{$\mathcal{M}_3$}
\psfrag{M4}[c][c][0.75]{$\mathcal{M}_4$}
\psfrag{M5}[c][c][0.75]{$\mathcal{M}_5$}
\psfrag{M6}[c][c][0.75]{$\mathcal{M}_6$}
\includegraphics[width=3.4 in]{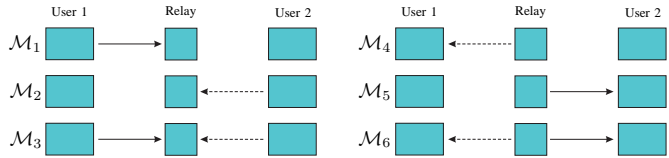}
\caption{The six possible transmission modes in the considered bidirectional relay network.}
\label{FigModes}
\end{figure}
In this section, we first describe the channel model. Then,
we provide the achievable rates for the six possible transmission modes.

\subsection{Channel Model}
We consider a simple network in which user 1 and user 2
exchange information with the help of a relay node as shown
in Fig. 1. We assume that there is no direct link between user
1 and user 2, and thus, user 1 and user 2 communicate with
each other only through the relay node. We assume that all
three nodes in the network are half-duplex. Furthermore, we
assume that time is divided into slots of equal length and that
each node transmits codewords which span one time slot or a
fraction of a time slot as will be explained later. We assume
that the user-to-relay and relay-to-user channels are impaired
by AWGN with unit variance and block fading, i.e., the channel
coefficients are constant during one time slot and change from
one time slot to the next. Moreover, in each time slot, the
channel coefficients are assumed to be reciprocal such that the
user 1-to-relay and the user 2-to-relay channels are identical to
the relay-to-user 1 and relay-to-user 2 channels, respectively. Let $h_1(i)$ and $h_2(i)$ denote the channel coefficients between user 1 and the relay and between user 2 and the relay in the $i$-th time slot, respectively. Furthermore, let $S_1(i)=|h_1(i)|^2$ and $S_2(i)=|h_2(i)|^2$  denote the squares  of the channel coefficient amplitudes in the $i$-th time slot. $S_1(i)$ and $S_2(i)$   are assumed to be ergodic and stationary random processes with means $\Omega_1=E\{S_1\}$ and  $\Omega_2=E\{S_2\}$\footnote{In this paper, we drop time index $i$ in expectations for notational simplicity.}, respectively, where $E\{\cdot\}$ denotes expectation.  Since the noise is AWGN, in order to achieve the capacity of each mode, nodes have to transmit
Gaussian distributed codewords. Therefore,  the transmitted codewords of user 1, user 2, and the relay are comprised of symbols which are Gaussian distributed random variables with variances $P_1(i), P_2(i)$, and $P_r(i)$, respectively, where $P_j(i)$ is the transmit power of node $j\in\{1,2,r\}$ in the $i$-th time slot. For ease of notation, we define $C(x)\triangleq \log_2(1+x)$.  In the following, we describe the transmission modes and their achievable rates.

\subsection{Transmission Modes and Their Achievable  Rates}

In the considered bidirectional relay network only six transmission modes are possible, cf. Fig. \ref{FigModes}. The six possible transmission modes are denoted by  ${\cal M}_1,...,{\cal M}_6$, and  $R_{jj'}(i)\geq 0, \,\,j,j'\in\{1,2,r\}$, denotes the transmission rate from node $j$ to node $j'$ in the $i$-th time slot. Let $B_1$ and $B_2$ denote two infinite-size buffers at the relay in which the received information from user 1 and user 2 is stored, respectively. Moreover, $Q_j(i), \,\, j\in\{1,2\}$, denotes the amount of normalized information in bits/symbol available in buffer $B_j$ in the $i$-th time slot. Using this notation, the transmission modes and their respective rates are presented in the following:

${\cal M}_1$: User 1 transmits to the relay and user 2 is silent. In this mode,     the maximum rate from user 1 to the relay in  the $i$-th time slot is given by $R_{1r}(i)= C_{1r}(i)$, where $C_{1r}(i)=C(P_1(i)S_1(i))$. The relay decodes this information and stores it in buffer $B_1$. Therefore, the amount of information in buffer $B_1$ increases to $Q_1(i)=Q_1(i-1)+R_{1r}(i)$.

${\cal M}_2$: User 2 transmits to the relay and user 1 is silent. In this mode,    the maximum rate from user 2 to the relay in  the $i$-th time slot is given by $R_{2r}(i)= C_{2r}(i)$, where $C_{2r}(i)=C(P_2(i)S_2(i))$.  The relay decodes this information and stores it in buffer $B_2$. Therefore, the amount of information in buffer $B_2$ increases to $Q_2(i)=Q_2(i-1)+R_{2r}(i)$.

${\cal M}_3$: Both users 1 and 2 transmit to the relay  simultaneously. For this mode, we assume that multiple access transmission is used, see \cite{Cover}. Thereby, the maximum achievable sum rate in the $i$-th time slot is given by $R_{1r}(i)+R_{2r}(i)= C_{r}(i)$, where $C_{r}(i)=C(P_1(i)S_1(i)+P_2(i)S_2(i))$.
 Since user 1 and user 2 transmit independent messages, the sum rate, $C_r(i)$, can be decomposed into two rates, one from user 1 to the relay and the other one from user 2 to the relay. Moreover, these two capacity rates can be achieved via time sharing and successive interference cancelation. Thereby, in the first $0\leq t(i)\leq 1$ fraction of the $i$-th time slot, the relay first decodes the
codeword received from user 2 and considers the signal from user 1 as noise. Then, the relay subtracts the signal received from user 2 from the received signal and decodes the
 codeword received from user 1. A similar procedure is performed in the remaining $1-t(i)$ fraction of the $i$-th time slot but now the relay first
decodes the codeword received from user 1 and treats the signal of user 2 as noise, and then decodes the codeword received from user 2. Therefore, for a given $t(i)$, we decompose $C_r(i)$ as $C_r(i)=C_{12r}(i)+C_{21r}(i)$ and  the maximum rates from users 1 and 2 to the relay in the $i$-th time slot are $R_{1r}(i) = C_{12r}(i)$ and $R_{2r}(i)=C_{21r}(i)$, respectively. $C_{12r}(i)$ and $C_{21r}(i)$ are given by
\begin{subequations}
\begin{align}\label{Cr}
  C_{12r}(i)\hspace{-0.5mm}&=\hspace{-0.5mm}t(i) C\left(P_1(i)S_{1}(i)\right)\hspace{-1mm} + \hspace{-1mm}(1-t(i)) C\left(\frac{P_1(i)S_{1}(i)}{1+P_2(i)S_{2}(i)}\right)  \IEEEeqnarraynumspace\IEEEyesnumber\IEEEyessubnumber \\
   C_{21r}(i)\hspace{-0.5mm}&=\hspace{-0.5mm}(1-t(i)) C\left(P_2(i)S_{2}(i)\right) \hspace{-1mm} +  \hspace{-1mm} t(i) C\left(\frac{P_2(i)S_{2}(i)}{1+P_1(i)S_{1}(i)}\right) \IEEEyessubnumber
\end{align}
\end{subequations}
The relay decodes the information received from user 1  and user 2  and stores it in its buffers $B_1$ and $B_2$, respectively. Therefore, the amounts of information in buffers $B_1$ and $B_2$ increase to $Q_1(i)=Q_1(i-1)+R_{1r}(i)$ and $Q_2(i)=Q_2(i-1)+R_{2r}(i)$, respectively.

${\cal M}_4$: The relay transmits the information received  from user 2 to user 1. Specifically, the relay extracts the information from buffer $B_2$, encodes it into a codeword, and transmits it to user 1. Therefore, the transmission rate from the relay to user 1 in the $i$-th time slot is limited by both the capacity of the relay-to-user 1 channel and the amount of  information  stored in buffer $B_2$. Thus, the maximum transmission rate from the relay to user 1 is given by $R_{r1}(i) \Equal \min\{C_{r1}(i),Q_2(i-1)\}$, where $C_{r1}(i)=C(P_r(i)S_1(i))$. Therefore, the amount of information in buffer $B_2$ decreases to $Q_2(i)\Equal Q_2(i\Minus 1)\Minus R_{r1}(i)$.

${\cal M}_5$: This mode is identical to ${\cal M}_4$ with user 1 and 2 switching places. The maximum transmission rate from the relay to user 2 is given by $R_{r2}(i) = \min\{C_{r2}(i),Q_1(i-1)\}$, where $C_{r2}(i)=C(P_r(i)S_2(i))$ and the amount of information in buffer $B_1$ decreases to $Q_1(i)\Equal Q_1(i\Minus 1)\Minus R_{r2}(i)$.

${\cal M}_6$: The relay broadcasts to both user 1 and user 2 the information received from user 2 and user 1, respectively. Specifically, the relay extracts the information intended for user 2 from buffer $B_1$ and  the information intended for user 1 from buffer $B_2$. Then, based on the scheme in \cite{BocheIT}, it constructs a superimposed codeword which contains the information from both users and broadcasts it to both users. Thus, in the $i$-th time slot, the maximum rates from the relay to users 1 and 2 are given by  $R_{r1}(i) \hspace{-1mm}=\hspace{-1mm} \min\hspace{-0.2mm}\{C_{r1}(i),\hspace{-0.5mm}Q_2(i-1)\}$ and $R_{r2}(i)=\min\{C_{r2}(i),Q_1(i-1)\}$, respectively. Therefore, the amounts of information in buffers $B_1$ and $B_2$ decrease to $Q_1(i)\Equal Q_1(i\Minus 1)\Minus R_{r2}(i)$ and $Q_2(i)\Equal Q_2(i\Minus 1)\Minus R_{r1}(i)$, respectively.

%
Our aim is to develop an  optimal mode selection and power allocation policy which in each time slot selects one of the six   transmission modes, ${\cal M}_1,...,{\cal M}_6$, and allocates the optimal powers to the transmitting nodes of the selected mode  such that the average sum rate of both users is maximized.
To this end,  we introduce six binary variables,  $q_k(i) \in\{0,1\}, \,\,k=1,...,6$, where  $q_k(i)$ indicates whether or not transmission mode  $\mathcal{M}_k$ is selected in the $i$-th time slot. In particular, $q_k(i)=1$ if  mode $\mathcal{M}_k$ is selected and $q_k(i)=0$ if it is not selected in the $i$-th time slot. Furthermore, since in each time slot  only one of the six transmission modes can be selected, only one of the mode selection variables is equal to one and the others are zero, i.e., $\mathop \sum_{k = 1}^6 q_k(i)=1$ holds.

 In the proposed framework, we assume that all nodes have full knowledge of the CSI of both links. Thus, based on the CSI and the proposed protocol, cf. Theorem~\ref{AdaptProt}, each node is able to individually decide which transmission mode is selected and adapt its transmission strategy accordingly.

\section{Joint Mode Selection and Power Allocation}\label{AdapMod}

In this section, we first investigate the achievable average sum rate  of the network. Then, we formulate a maximization problem whose solution is the sum rate maximizing protocol.

\subsection{Achievable Average Sum Rate}

We assume that user 1 and user 2 always have enough information to send in all time slots and that the number of time slots, $N$, satisfies $N\to \infty$. Therefore, using $q_k(i)$, the user 1-to-relay, user 2-to-relay, relay-to-user 1, and relay-to-user 2 average transmission rates, denoted by $\bar{R}_{1r}$, $\bar{R}_{2r}$, $\bar{R}_{r1}$, and $\bar{R}_{r2}$, respectively, are obtained as
\begin{IEEEeqnarray}{lll}\label{RatReg123}
    &&\bar{R}_{1r} = \underset{N\to \infty}{\lim} \frac{1}{N}\mathop \sum \limits_{i = 1}^N \left[ q_1(i)C_{1r}(i)+q_3(i)C_{12r}(i)\right] \IEEEyesnumber\IEEEyessubnumber \\
		&&\bar{R}_{2r} =  \underset{N\to \infty}{\lim} \frac{1}{N}\mathop \sum \limits_{i = 1}^N \left[ q_2(i)C_{2r}(i)+q_3(i)C_{21r}(i)\right] \IEEEyessubnumber\\
    &&\bar{R}_{r1}\hspace{-1mm}  =\hspace{-1mm} \underset{N\to \infty}{\lim}\frac{1}{N}\hspace{-1mm} \mathop \sum \limits_{i = 1}^N \left[ q_4(i)\hspace{-0.5mm} +\hspace{-0.5mm} q_6(i)\right]\hspace{-0.5mm} \min\{C_{r1}(i),Q_2(i-1)\} \IEEEeqnarraynumspace\IEEEyessubnumber \\
		&&\bar{R}_{r2}\hspace{-1mm}  =\hspace{-1mm}   \underset{N\to \infty}{\lim}\frac{1}{N}\hspace{-1mm} \mathop \sum \limits_{i = 1}^N \left[ q_5(i)\hspace{-0.5mm} +\hspace{-0.5mm} q_6(i)\right]\hspace{-0.5mm} \min\{C_{r2}(i),Q_1(i-1)\}. \IEEEeqnarraynumspace \IEEEyessubnumber
\end{IEEEeqnarray}
The average rate from user 1 to user 2 is the average rate that user 2 receives from the relay, i.e.,  $\bar{R}_{r2}$. Similarly,   the average rate from user 2 to user 1 is the average rate that user 1 receives from the relay, i.e.,  $\bar{R}_{r1}$.
 In the following theorem, we introduce a useful condition for the queues in the buffers of the relay  leading to the optimal mode selection and power allocation policy.

\begin{theorem}[Optimal Queue Condition]\label{Queue}
\normalfont
The maximum average sum rate, $\bar{R}_{r1}+\bar{R}_{r2}$, for the considered bidirectional relay network is obtained when  the queues in the buffers $B_1$ and $B_2$   at the relay are at the edge of non-absorbtion. More precisely, the following conditions must hold for the maximum sum rate
\begin{IEEEeqnarray}{lll}\label{RatRegApp456-buffer}
\bar{R}_{1r}=\bar{R}_{r2} = \underset{N\to \infty}{\lim}\frac{1}{N}\mathop \sum \limits_{i = 1}^N \left[ q_5(i)+q_6(i)\right]C_{r2}(i)  \IEEEyesnumber\IEEEeqnarraynumspace \IEEEyessubnumber \\
\bar{R}_{2r}=\bar{R}_{r1}= \underset{N\to \infty}{\lim}\frac{1}{N}\mathop \sum \limits_{i = 1}^N \left[ q_4(i)+q_6(i)\right]C_{r1}(i) \IEEEeqnarraynumspace\IEEEyessubnumber
\end{IEEEeqnarray}
where $\bar{R}_{1r}$ and $\bar{R}_{2r}$ are given by (\ref{RatReg123}a) and (\ref{RatReg123}b), respectively.
\end{theorem}

\begin{proof}
Please refer to \cite[Appendix A]{EUSIPCO}.
\end{proof}
Using this theorem, in the following, we   derive the optimal transmission mode selection and power allocation policy.

\subsection{Optimal Protocol}

The available degrees of freedom in the considered network in each time slot are the mode selection variables,
the transmit powers of the nodes, and the time sharing variable for multiple access.
Herein, we formulate an optimization problem which gives the optimal values of $q_k(i)$, $P_j(i)$, and $t(i)$,
for $k=1,...,6$, $j=1,2,r$, and $\forall i$, such that the average sum rate of the users is maximized.
The optimization problem is  as follows
\begin{IEEEeqnarray}{Cll}\label{AdaptProb}
    {\underset{q_k(i),P_j(i),t(i),\,\, \forall i,j,k}{\mathrm{maximize}}}\,\, &\bar{R}_{1r}+\bar{R}_{2r} \nonumber \\
    \mathrm{subjected\,\, to} \,\, &\mathrm{C1}:\quad \bar{R}_{1r}=\bar{R}_{r2}  \nonumber \\
    &\mathrm{C2}:\quad \bar{R}_{2r}=\bar{R}_{r1} \nonumber \\
    &\mathrm{C3}:\quad \bar{P}_1+\bar{P}_2+\bar{P}_r \leq P_t \nonumber\\
		&\mathrm{C4}:\quad \sum\nolimits_{k = 1}^6 {q_k}\left( i \right) = 1, \,\, \forall i   \nonumber \\
    &\mathrm{C5}:\quad q_k(i) [1-q_k(i)] = 0, \,\, \forall i, k \nonumber \\
    &\mathrm{C6}:\quad P_j(i)\geq 0, \,\, \forall i, j \nonumber\\
    &\mathrm{C7}:\quad 0\leq t(i)\leq 1, \,\, \forall i \IEEEyesnumber
\end{IEEEeqnarray}
where $P_t$ is the total average power constraint of the nodes and $\bar{P}_1,\bar{P}_2$, and $\bar{P}_r$ denote the average powers
consumed by user 1, user 2, and the relay, respectively, and are given by
\begin{IEEEeqnarray}{lll}
    \bar{P}_1=\frac{1}{N} \mathop \sum \limits_{i = 1}^N (q_1(i)+q_3(i))P_1 (i) \IEEEyesnumber\IEEEyessubnumber \\
    \bar{P}_2=\frac{1}{N} \mathop \sum \limits_{i = 1}^N (q_2(i)+q_3(i))P_2 (i) \IEEEyessubnumber \\
    \bar{P}_r=\frac{1}{N} \mathop \sum \limits_{i = 1}^N (q_4(i)+q_5(i)+q_6(i))P_r (i). \IEEEyessubnumber
    \end{IEEEeqnarray}
In the optimization problem given in (\ref{AdaptProb}), constraints $\mathrm{C1}$ and $\mathrm{C2}$ are the conditions for sum rate maximization introduced in Theorem \ref{Queue}. Constraints $\mathrm{C3}$ and $\mathrm{C6}$ are the average total transmit power constraint and the power non-negativity constraint, respectively. Moreover, constraints $\mathrm{C4}$ and $\mathrm{C5}$ guarantee that only one of the transmission modes is selected in each time slot, and constraint $\mathrm{C7}$ specifies the acceptable interval for the time sharing variable $t(i)$. Furthermore, we maximize  $\bar{R}_{1r}+\bar{R}_{2r}$ since, according to Theorem 1 (and constraints $\mathrm{C1}$ and $\mathrm{C2}$), $\bar{R}_{1r}=\bar{R}_{r2}$ and  $\bar{R}_{2r}=\bar{R}_{r1}$ hold.

In the following Theorem, we introduce a protocol which
achieves the maximum sum rate.

\begin{theorem}[Mode Selection and Power Allocation Policy]\label{AdaptProt}
 Assuming $N\to \infty$, the optimal mode selection and power allocation policy which maximizes the sum rate of the considered three-node half-duplex bidirectional relay network with AWGN and block fading is given by
\begin{IEEEeqnarray}{lll}\label{SelecCrit}
q_{k^*}(i)=
   \begin{cases}
     1, & \textrm{if }k^*= {\underset{k=1,2,3,6}{\arg \, \max}} \{\Lambda_k(i)\} \\
     0, &\mathrm{otherwise}
\end{cases}
 \IEEEyesnumber
\end{IEEEeqnarray}
where $\Lambda_k(i)$ is referred to as \textit{selection metric} and is given by
\begin{IEEEeqnarray}{lll}\label{SelecMet}
    \Lambda_1(i) = (1-\mu_1)C_{1r}(i)  - \gamma  P_1(i) \big |_{P_1(i)=P_1^{\mathcal{M}_1}(i)} \IEEEyesnumber\IEEEyessubnumber  \\
    \Lambda_2(i) =  (1-\mu_2)C_{2r}(i)  - \gamma  P_2(i)\big |_{P_2(i)=P_2^{\mathcal{M}_2}(i)} \IEEEyessubnumber  \\
    \Lambda_3(i) = (1-\mu_1)C_{12r}(i)+(1-\mu_2)C_{21r}(i)  \nonumber \\
 \qquad\,\,\,- \gamma ( P_1(i)+P_2(i))\Big |_{P_2(i)=P_2^{\mathcal{M}_3}(i)}^{P_1(i)=P_1^{\mathcal{M}_3}(i)}\quad\,\,\IEEEyessubnumber  \\
  \Lambda_6(i) = \mu_1 C_{r2}(i)+\mu_2 C_{r1}(i)  - \gamma  P_r(i)\big |_{P_r(i)=P_r^{\mathcal{M}_6}(i)}\IEEEyessubnumber
\end{IEEEeqnarray}
where $P_j^{\mathcal{M}_k}(i)$ denotes the optimal transmit power of node $j$ for  transmission mode $\mathcal{M}_k$ in the $i$-th time slot and is given by
\begin{IEEEeqnarray}{lll}\label{OptPower}
   P_1^{\mathcal{M}_1} (i) = \left[\frac{1-\mu_1}{\gamma \mathrm{ln}2}-\frac{1}{S_1(i)}\right]^+
\IEEEyesnumber \IEEEyessubnumber \\
P_2^{\mathcal{M}_2} (i) = \left[\frac{1-\mu_2}{\gamma \mathrm{ln}2}-\frac{1}{S_2(i)}\right]^+
\IEEEyessubnumber \\
P_1^{\mathcal{M}_3} (i) =
{\begin{cases}
\left[\frac{1-\mu_1}{\gamma \mathrm{ln}2}-\frac{\mu_1-\mu_2}{\gamma \mathrm{ln}2}\frac{1}{\frac{S_1(i)}{S_2(i)}-1}\right]^+, \mathrm{if} \,\, \Omega_1\geq \Omega_2 \\
\left[\frac{\mu_1-\mu_2}{\gamma \mathrm{ln}2}\frac{1}{\frac{S_1(i)}{S_2(i)}-1} - \frac{1}{S_1(i)}\right]^+, \mathrm{otherwise}
\end{cases}} \IEEEyessubnumber \\ 
P_2^{\mathcal{M}_3} (i) =
{\begin{cases}
\left[\frac{\mu_1-\mu_2}{\gamma \mathrm{ln}2}\frac{1}{1-\frac{S_2(i)}{S_1(i)}} - \frac{1}{S_2(i)}\right]^+, \mathrm{if} \,\, \Omega_1 \geq \Omega_2 \\
\left[\frac{1-\mu_2}{\gamma \mathrm{ln}2}-\frac{\mu_1-\mu_2}{\gamma \mathrm{ln}2}\frac{1}{1-\frac{S_2(i)}{S_1(i)}}\right]^+, \mathrm{otherwise}
\end{cases}}  \IEEEyessubnumber \\
P_r^{\mathcal{M}_6} (i) = \left [ \frac{-b+\sqrt{b^2-4ac}}{2a} \right]^+ \IEEEyessubnumber
\end{IEEEeqnarray}
where $[x]^+=\max\{x,0\}$, $a=\gamma \mathrm{ln}2\times S_1(i)S_2(i), b= \gamma \mathrm{ln}2 \times (S_1(i)+S_2(i)) -
(\mu_1+\mu_2)S_1(i)S_2(i)$, and $c=\gamma \mathrm{ln}2 - \mu_1 S_2(i) - \mu_2 S_1(i)$. The thresholds $\mu_1$ and
$\mu_2$ are chosen such that constraints $\mathrm{C1}$ and $\mathrm{C2}$ in (\ref{AdaptProb}) hold and threshold
$\gamma$ is chosen such that the total  average transmit power satisfies $\mathrm{C3}$ in (\ref{AdaptProb}).
The optimal value of $t(i)$ in $C_{12r}(i)$ and $C_{21r}(i)$ is given by
\begin{IEEEeqnarray}{lll}
   t^*(i) = {\begin{cases} 0, & \Omega_1 \geq \Omega_2 \\
    1, & \Omega_1 < \Omega_2
\end{cases}}
\end{IEEEeqnarray}
\end{theorem}

\begin{proof}
Please refer to Appendix \ref{AppKKT}.
\end{proof}

We note that the optimal solution utilizes neither modes $\mathcal{M}_4$ and $\mathcal{M}_5$ nor time sharing for any channel statistics and channel realizations.

\begin{remark1}
The mode selection metric $\Lambda_k(i)$ introduced in (\ref{SelecMet}) has two parts. The first part is the instantaneous capacity of mode $\mathcal{M}_k$, and the second part is the allocated power with negative sign. The capacity and the power terms are linked via thresholds $\mu_1$ and/or $\mu_2$  and $\gamma$. We note that thresholds $\mu_1$, $\mu_2$, and $\gamma$ depend only on the long term statistics of the channels. Hence, these thresholds can be obtained offline and used as long as the channel statistics remain unchanged. To find the optimal values for the thresholds $\mu_1$, $\mu_2$, and $\gamma$, we need a three-dimensional search, where  $\mu_1,\mu_2 \in (0\,\,1)$  and $\gamma>0$.
\end{remark1}

\begin{remark2}
Adaptive mode selection for bidirectional relay
networks under the assumption that the powers of the nodes are
fixed is considered in \cite{EUSIPCO}. Based on the average and instantaneous qualities of the links, all of the six possible transmission modes are selected in the protocol in \cite{EUSIPCO}. However, in the proposed protocol, modes
$\mathcal{M}_4$ and $\mathcal{M}_5$ are not selected at all. Moreover, the protocol
in \cite{EUSIPCO} utilizes a coin flip for implementation. Therefore, a
central node must decide which transmission mode is selected
in the next time slot. However, in the proposed protocol, all  nodes  can  find  the  optimal  mode  and
powers based on the full CSI.
\end{remark2}

\section{Simulation Results}\label{SimRes}

In this section, we evaluate the average sum rate achievable with the proposed protocol in the considered bidirectional relay network in Rayleigh fading. Thus, channel gains $S_1(i)$ and $S_2(i)$ follow exponential distributions with means $\Omega_1$ and $\Omega_2$, respectively. All of the presented results were obtained for $\Omega_2=1$ and $N=10^4$   time slots.

In Fig.~\ref{Mixed}, we  illustrate the maximum achievable sum rate obtained with the proposed protocol as a function of the total  average transmit power $P_t$. In this figure, to have
a better resolution for the sum rate at low and high $P_t$, we
show the sum rate for both log scale and linear scale $y$-axes,
respectively. The lines without markers in Fig.~\ref{Mixed} represent the achieved sum rates with the proposed protocol for $\Omega_1=1,2,5$. We observe that as the quality of the user 1-to-relay link increases (i.e., $\Omega_1$ increases), the sum rate increases too. However, for large $\Omega_1$, the bottleneck link is the relay-to-user 2 link, and since it is fixed, the sum rate saturates.

\begin{figure}
\centering
\psfrag{Y}[c][t][0.75]{$\text{Sum Rate}$}
\psfrag{X}[c][b][0.75]{$\text{Total Power Budget (dB)}$}
\psfrag{L1}[l][c][0.45]{$\text{Proposed Protocol}\,(\Omega_1=5)$}
\psfrag{L2}[l][c][0.45]{$\text{Proposed Protocol}\,(\Omega_1=2)$}
\psfrag{L3}[l][c][0.45]{$\text{Proposed Protocol}\,(\Omega_1=1)$}
\psfrag{L4}[l][c][0.45]{$\text{Selection Policy in [11]}\,(\Omega_1=1)$}
\psfrag{L5}[l][c][0.45]{$\text{Selection Policy in [8]}\,(\Omega_1=1)$}
\psfrag{L6}[l][c][0.45]{$\text{TDBC with PA}\,(\Omega_1=1)$}
\psfrag{L7}[l][c][0.45]{$\text{TDBC without PA}\,(\Omega_1=1)$}
\psfrag{00}[c][c][0.6]{$-20$}
\psfrag{01}[c][c][0.6]{$-15$}
\psfrag{02}[c][c][0.6]{$-10$}
\psfrag{03}[c][c][0.6]{$-5$}
\psfrag{04}[c][c][0.6]{$0$}
\psfrag{05}[c][c][0.6]{$5$}
\psfrag{06}[c][c][0.6]{$10$}
\psfrag{07}[c][c][0.6]{$15$}
\psfrag{08}[c][c][0.6]{$20$}
\psfrag{09}[c][c][0.6]{$10^{-3}\quad$}
\psfrag{10}[c][c][0.6]{$10^{-2}\quad$}
\psfrag{99}[c][c][0.6]{$10^{-1}\quad$}
\psfrag{12}[c][c][0.6]{$10^{0}\quad$}
\psfrag{13}[c][c][0.6]{$10^{1}\quad$}
\psfrag{14}[c][c][0.30]{$0$}
\psfrag{15}[c][c][0.30]{$0.5$}
\psfrag{16}[c][c][0.30]{$1$}
\psfrag{17}[c][c][0.30]{$1.5$}
\psfrag{18}[c][c][0.30]{$2$}
\psfrag{19}[c][c][0.30]{$2.5$}
\psfrag{20}[c][c][0.30]{$3$}
\psfrag{21}[c][c][0.30]{$3.5$}
\psfrag{22}[c][c][0.30]{$4$}
\psfrag{23}[c][c][0.30]{$4.5$}
\psfrag{24}[c][c][0.30]{$5$}
\psfrag{25}[c][c][0.30]{$-20$}
\psfrag{26}[c][c][0.30]{$-15$}
\psfrag{27}[c][c][0.30]{$-10$}
\psfrag{28}[c][c][0.30]{$-5$}
\psfrag{29}[c][c][0.30]{$0$}
\psfrag{30}[c][c][0.30]{$5$}
\psfrag{31}[c][c][0.30]{$10$}
\psfrag{32}[c][c][0.30]{$15$}
\psfrag{33}[c][c][0.30]{$20$}
\includegraphics[width=3.7 in]{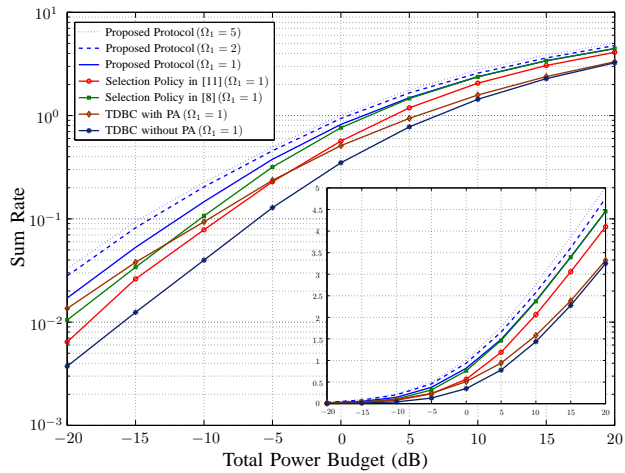}
\caption{Maximum sum rate versus $P_t$ for different protocols.}
\label{Mixed}
\end{figure}

As performance benchmarks, we consider in Fig.~\ref{Mixed} the sum rates of the TDBC protocol with and without power allocation \cite{TDBC} and the buffer-aided protocols presented in \cite{PopovskiLetter} and \cite{EUSIPCO},
respectively. For clarity, for the benchmark schemes, we only
show the sum rates for $\Omega_1=\Omega_2$. For the TDBC protocol
without power allocation and the protocol in \cite{PopovskiLetter}, all nodes
transmit with equal powers, i.e., $P_1 = P_2 = P_r = P_t$. For the
buffer-aided protocol in \cite{EUSIPCO}, we adopt  $P_1 = P_2 = P_r = P$
and $P$ is chosen such that the average total power consumed by all
nodes is $P_t$. We note that since $\Omega_1=\Omega_2$ and $P_1 = P_2$, the
protocol in \cite{EUSIPCO} only selects modes $\mathcal{M}_3$ and $\mathcal{M}_6$. Moreover, since $\Omega_1=\Omega_2$, we obtain $\mu_1=\mu_2$ in the proposed protocol.
Thus, considering the optimal power allocation in (\ref{OptPower}c) and
(\ref{OptPower}d), we obtain that either $P_1^{\mathcal{M}_3}(i)$ or $P_2^{\mathcal{M}_3}(i)$ is zero. Therefore, for the chosen parameters,
only modes $\mathcal{M}_1,\mathcal{M}_2$, and $\mathcal{M}_6$ are selected, i.e., the same modes as used in \cite{PopovskiLetter}\footnote{We note that although the protocol in \cite{PopovskiLetter} outperforms the protocol in \cite{EUSIPCO} if the sum rate is plotted as a function of total transmit power as is done in Fig.~\ref{Mixed}, the protocol in \cite{EUSIPCO} is optimal for given fixed node transmit powers.}. Hence, we can see how much gain we obtain due to the adaptive power allocation
by comparing our result with the results for the protocol in
\cite{PopovskiLetter}. On the other hand, the gain due to the adaptive mode
selection can be evaluated by comparing the sum rate of
the proposed protocol with the result for the TDBC protocol
with power allocation. From the comparison in Fig. \ref{Mixed}, we
observe that for high $P_t$, a considerable gain is obtained by
the protocols with adaptive mode selection (ours and that in
\cite{PopovskiLetter}) compared to the TDBC protocol which does not apply
adaptive mode selection (around $6$ dB gain). However, for high
$P_t$, power allocation is less beneficial, and therefore, the
sum rates obtained with the proposed protocol and that
in \cite{PopovskiLetter} converge. On the other hand, for low $P_t$, optimal
power allocation is crucial and, therefore, a considerable gain
is achieved by the protocols with adaptive power allocation
(ours and TDBC with power allocation).


\section{Conclusion}\label{Conclusion}
We have derived the maximum sum rate of the three-node
half-duplex bidirectional buffer-aided relay network with
fading links. The protocol which achieves the maximum sum
rate jointly optimizes the selection of the transmission mode
and the transmit powers of the nodes. The proposed optimal
mode selection and power allocation protocol requires the
instantaneous CSI of the involved links in each time slot and
their long-term statistics. Simulation results confirmed that the
proposed selection policy outperforms existing protocols in
terms of average sum rate.

\appendices

\section{Proof of Theorem 2 (Mode Selection Protocol)}
\label{AppKKT}

In this appendix, we solve the optimization problem given in (\ref{AdaptProb}). We first relax the binary condition for $q_k(i)$, i.e., $q_k(i)[1-q_k(i)]=0$, to $0\leq q_k(i)\leq 1$, and later in Appendix \ref{AppBinRelax}, we prove  that the binary relaxation does not affect the maximum average sum rate. In the following, we investigate the Karush-Kuhn-Tucker (KKT) necessary conditions \cite{Boyd} for the relaxed optimization problem and  show that the necessary conditions result in a unique sum rate and thus the solution is optimal.

To simplify the usage of the KKT conditions, we formulate a minimization problem equivalent to the relaxed maximization
problem in (\ref{AdaptProb}) as follows
\begin{IEEEeqnarray}{Cll}\label{AdaptProbMin}
    {\underset{q_k(i),P_j(i),t(i),\,\,\forall i,k,j}{\mathrm{minimize}}}\,\, &-(\bar{R}_{1r}+\bar{R}_{2r}) \nonumber \\
    \mathrm{subject\,\, to} \,\, &\mathrm{C1}:\quad \bar{R}_{1r}-\bar{R}_{r2}=0  \nonumber \\
    &\mathrm{C2}:\quad \bar{R}_{2r}-\bar{R}_{r1}=0 \nonumber \\
&\mathrm{C3}:\quad \bar{P}_1+\bar{P}_2+\bar{P}_r - P_t\leq 0  \nonumber\\
		&\mathrm{C4}:\quad \mathop \sum \nolimits_{k = 1}^6 {q_k}\left( i \right) - 1 =0, \,\, \forall i   \nonumber \\
    &\mathrm{C5}:\quad q_k(i)-1 \leq 0, \,\, \forall i, k \nonumber \\
&\mathrm{C6}:\quad -q_k(i) \leq 0, \,\, \forall i, k \nonumber \\
&\mathrm{C7}:\quad -P_j(i) \leq 0, \,\, \forall i,k \nonumber \\
    &\mathrm{C8}:\quad t(i)-1 \leq 0, \,\, \forall i \nonumber \\
&\mathrm{C9}:\quad -t(i) \leq 0, \,\, \forall i. \IEEEyesnumber
\end{IEEEeqnarray}
The Lagrangian function for the above optimization problem is provided in (\ref{KKT Function}) at the top of the next page
where $\mu_1,\mu_2,\gamma,\lambda(i),\alpha_k(i),\beta_k(i),
\nu_j(i),\phi_1(i)$, and $\phi_2(i)$ are the Lagrange multipliers corresponding to constraints $\mathrm{C1,C2,C3,C4,C5,C6,C7,C8}$, and $\mathrm{C9}$, respectively. The KKT conditions include the following:

\begin{figure*}[!t]
\normalsize
\begin{IEEEeqnarray}{l}\label{KKT Function}
   \underset{\mathrm{for}\,\, \forall i,k,j,l}{\mathcal{L}(q_k(i),P_j(i),t(i),\mu_l,\gamma,\lambda(i),\alpha_k(i),\beta_k(i),\phi_l(i))}  = \nonumber \\ 
\qquad \,\, - (\bar{R}_{1r}+\bar{R}_{2r}) + \mu_1(\bar{R}_{1r}-\bar{R}_{r2}) + \mu_2(\bar{R}_{2r}-\bar{R}_{r1}) + \gamma \left(\bar{P}_1+\bar{P}_2+\bar{P}_r - P_t\right) \nonumber \\
   \qquad \,\, + \mathop \sum \limits_{i = 1}^N \lambda \left( i \right)\left( {\mathop \sum \limits_{k = 1}^6 {q_k}\left( i \right) - 1} \right)
+ \mathop \sum \limits_{i = 1}^N \mathop \sum \limits_{k = 1}^6 {\alpha _k}\left( i \right)\left( {{q_k}\left( i \right) - 1} \right) - \mathop \sum \limits_{i = 1}^N \mathop \sum \limits_{k = 1}^6 {\beta _k}\left( i \right){q_k}\left( i \right) \nonumber \\ \qquad \,\,
-\mathop \sum \limits_{i = 1}^N \left[ \nu_1(i)P_1(i) + \nu_2(i)P_2(i) + \nu_3(i)P_3(i) \right] + \mathop \sum \limits_{i = 1}^N \phi_1(i) (t(i)-1) - \mathop \sum \limits_{i = 1}^N \phi_0(i) t(i) \IEEEeqnarraynumspace \IEEEyesnumber
\end{IEEEeqnarray}
\hrulefill

\vspace*{4pt}
\end{figure*}

\noindent
\textbf{1)} Stationary condition: The differentiation of the Lagrangian function with respect to the
primal variables, $q_k(i),P_j(i)$, and $t(i),\,\,\forall i,j,k$, is zero for the optimal solution, i.e.,
\begin{IEEEeqnarray}{CCCl}\label{Stationary Condition}
    \frac{\partial\mathcal{L}}{\partial q_k(i)} &=& 0, \quad &\forall i,k \IEEEyesnumber\IEEEyessubnumber\\
\frac{\partial\mathcal{L}}{\partial P_j(i)} &=& 0, \quad &\forall i,j \IEEEyessubnumber \\
\frac{\partial\mathcal{L}}{\partial t(i)} &=& 0, \quad &\forall i.\IEEEyessubnumber
\end{IEEEeqnarray}

\noindent
\textbf{2)} Primal feasibility condition: The optimal solution has to satisfy the constraints of the primal problem in (\ref{AdaptProbMin}).

\noindent
\textbf{3)} Dual feasibility condition: The Lagrange multipliers for the inequality constraints have to be non-negative, i.e.,
\begin{IEEEeqnarray}{lll}\label{Dual Feasibility Condition}
            \alpha_k(i)\geq 0, \quad &\forall i,k \IEEEyesnumber\IEEEyessubnumber\\
         \beta_k(i)\geq 0,\quad &\forall i,k \IEEEyessubnumber\\
          \gamma \geq 0,   &  \IEEEyessubnumber \\
           \nu_j(i) \geq0,   & \forall i,j \IEEEyessubnumber \\
             \phi_l(i) \geq0,   & \forall i,l. \IEEEyessubnumber
\end{IEEEeqnarray}

\noindent
\textbf{4)} Complementary slackness: If an inequality is inactive, i.e., the optimal solution is in the interior of the corresponding set, the corresponding Lagrange multiplier is zero. Thus, we obtain
\begin{IEEEeqnarray}{lll}\label{Complementary Slackness}
    {\alpha _k}\left( i \right)\left( {{q_k}\left( i \right) - 1} \right)=0,\quad  &\forall i,k \IEEEyesnumber\IEEEyessubnumber\\
    {\beta _k}\left( i \right){q_k}\left( i \right)=0,  &\forall i,k \IEEEyessubnumber\\
   \gamma (\bar{P}_1+\bar{P}_2+\bar{P}_r - P_t) = 0 \,\, & \IEEEyessubnumber\\
    \nu_j(i) P_j(i) = 0, &\forall i,j \IEEEyessubnumber\\
    \phi_1(i) (t(i)-1) = 0, &\forall i \IEEEyessubnumber\\
		\phi_0(i) t(i) = 0, &\forall i.\IEEEyessubnumber
\end{IEEEeqnarray}
A common approach to find a set of primal variables, i.e., $q_k(i), P_j(i),t(i),\,\,\forall i,j,k$ and Lagrange multipliers, i.e., $\mu_1,\mu_2,\gamma,\lambda(i),\alpha_k(i),\beta_k(i),\nu_j(i),\phi_l(i),\,\,
\forall i,k,l$, which satisfy the KKT conditions is to start with the complementary slackness conditions and see if the inequalities are active or not. Combining these results with the primal feasibility and dual feasibility conditions, we obtain various possibilities. Then, from these possibilities, we obtain one or more candidate solutions from the stationary conditions and the optimal solution is surely one of these candidates. In the following subsections, with this approach, we find the optimal values of $q_k^*(i),P_j^*(i),$ and $t^*(i),\,\,\forall i,j,k$.

\subsection{Optimal $q_k^*(i)$}

In order to determine the optimal selection policy, $q_k^*(i)$, we must calculate the derivatives
in (\ref{Stationary Condition}a). This leads to
\begin{IEEEeqnarray}{lll}\label{Stationary Mode}
    \frac{\partial\mathcal{L}}{\partial q_1(i)} = - \frac{1}{N}(1\Minus\mu_1)C_{1r}(i)\Add \lambda(i)\Add \alpha_1(i)\Minus \beta_1(i) \nonumber  + \frac{1}{N}\gamma P_1(i)\Equal 0\, \IEEEeqnarraynumspace \IEEEyesnumber \IEEEyessubnumber \\
    \frac{\partial\mathcal{L}}{\partial q_2(i)} = -\frac{1}{N}(1\Minus\mu_2)C_{2r}(i)\Add\lambda(i)\Add\alpha_2(i)\Minus\beta_2(i) \nonumber  +\frac{1}{N}\gamma P_2(i)\Equal 0 \IEEEyessubnumber   \\
    \frac{\partial\mathcal{L}}{\partial q_3(i)} = -\frac{1}{N}[(1\Minus\mu_1)C_{12r}(i)\Add(1\Minus\mu_2)C_{21r}(i)]\Add\lambda(i) +\alpha_3(i)\Minus\beta_3(i)\Add\frac{1}{N}\gamma (P_1(i)\Add P_2(i))\Equal 0 \,\,  \IEEEyessubnumber \\
    \frac{\partial\mathcal{L}}{\partial q_4(i)} = - \frac{1}{N}\mu_2 C_{r1}(i)\Add\lambda(i)\Add\alpha_4(i)\Minus \beta_4(i)\Add\frac{1}{N}\gamma P_r(i)\Equal 0 \qquad \, \IEEEyessubnumber \\
    \frac{\partial\mathcal{L}}{\partial q_5(i)} = -\frac{1}{N}\mu_1 C_{r2}(i)\Add\lambda(i)\Add\alpha_5(i)\Minus\beta_5(i)\Add \frac{1}{N}\gamma P_r(i)\Equal 0 \qquad \, \IEEEyessubnumber \\
    \frac{\partial\mathcal{L}}{\partial q_6(i)} = -\frac{1}{N}[\mu_1 C_{r2}(i)\Add\mu_2 C_{r1}(i)]\Add\lambda(i)\nonumber +\alpha_6(i)\Minus\beta_6(i)\Add\frac{1}{N}\gamma P_r(i)\Equal 0. \IEEEyessubnumber
\end{IEEEeqnarray}
Without loss of generality, we first obtain the necessary condition for $q_1^*(i)=1$ and then generalize the result
to $q_k^*(i)=1,\,\,k=2,\dots,6$. If $q_k^*(i)=1$, from constraint $\mathrm{C4}$ in (\ref{AdaptProbMin}), the other
selection variables are zero, i.e., $q_k^*(i)=0,\,\,k=2,...,6$. Furthermore, from (\ref{Complementary Slackness}),
we obtain $\alpha_k(i)=0,\,\,k = 2,...,6$ and $ \beta_1(i)=0$. By substituting these values into (\ref{Stationary Mode}), we obtain
\begin{IEEEeqnarray}{lll}\label{MET}
    \lambda(i)+\alpha_1(i) = (1-\mu_1)C_{1r}(i) -\gamma P_1(i) \triangleq \Lambda_1(i) \IEEEyesnumber\IEEEyessubnumber  \\
    \lambda(i)-\beta_2(i) =  (1-\mu_2)C_{2r}(i) -\gamma P_2(i) \triangleq \Lambda_2(i)\IEEEyessubnumber  \\
   \lambda(i)-\beta_3(i) =  (1\Minus \mu_1)C_{12r}(i)\Add (1\Minus \mu_2)C_{21r}(i) -\gamma (P_1(i)\Add P_2(i)) \triangleq  \Lambda_3(i) \quad\,\,\,\, \IEEEyessubnumber \\
   \lambda(i)-\beta_4(i)  = \mu_2 C_{r1}(i) -\gamma P_r(i) \triangleq \Lambda_4(i) \IEEEyessubnumber\\
    \lambda(i)-\beta_5(i) = \mu_1 C_{r2}(i) -\gamma P_r(i) \triangleq \Lambda_5(i) \IEEEyessubnumber\\
    \lambda(i)-\beta_6(i) = \mu_1 C_{r2}(i)+\mu_2 C_{r1}(i) -\gamma P_r(i) \triangleq \Lambda_6(i),\qquad\IEEEyessubnumber
\end{IEEEeqnarray}
where $\Lambda_k(i)$ is referred to as selection metric. By subtracting (\ref{MET}a) from the rest of the equations in (\ref{MET}), we obtain
\begin{IEEEeqnarray}{rCl}\label{eq_2_1}
    \Lambda_1(i) - \Lambda_k(i) = \alpha_1(i)+\beta_k(i), \quad k=2,3,4,5,6. \IEEEyesnumber
\end{IEEEeqnarray}
From the dual feasibility conditions given in (\ref{Dual Feasibility Condition}a) and (\ref{Dual Feasibility Condition}b), we have $\alpha_k(i),\beta_k(i)\geq 0$. By inserting $\alpha_k(i),\beta_k(i)\geq 0$ in (\ref{eq_2_1}), we obtain the necessary condition for $q_1^*(i)=1$ as
\begin{IEEEeqnarray}{lll}
    \Lambda_1(i) \geq \max \left \{ \Lambda_2(i), \Lambda_3(i), \Lambda_4(i), \Lambda_5(i), \Lambda_6(i) \right \}. \IEEEyesnumber
\end{IEEEeqnarray}
Repeating the same procedure for $q_k^*(i)=1,\,\,k=2,\dots,6$, we obtain a necessary condition for selecting transmission mode $\mathcal{M}_{k^*}$ in the $i$-th time slot as follows
\begin{IEEEeqnarray}{lll}\label{OptMet}
   \Lambda_{k^*}(i) \geq {\underset{k\in\{1,\cdots,6\}}{\max}}\{\Lambda_{k}(i)\}, \IEEEyesnumber
\end{IEEEeqnarray}
where the Lagrange multipliers $\mu_1,\mu_2$, and $\gamma$ are chosen such that $\mathrm{C1,C2}$, and $\mathrm{C3}$ in (\ref{AdaptProbMin}) hold and the optimal value of $t(i)$ in $C_{12r}(i)$ and $C_{21r}(i)$ is obtained in the next subsection.
We note that if
the selection metrics are not equal in the $i$-th time slot, only one of
the modes satisfies (\ref{OptMet}). Therefore, the necessary conditions for  the  mode  selection  in  (\ref{OptMet})  is  sufficient.  Moreover,  in
Appendix \ref{AppBinRelax}, we prove that the probability that two selection
metrics are equal is zero due to the randomness of the time-
continuous channel gains. Therefore, the necessary condition
for selecting transmission mode $\mathcal{M}_k$ in (\ref{OptMet}) is in fact sufficient
and is the optimal selection policy.

\subsection{Optimal $P_j^*(i)$}

In order to determine the optimal $P_j(i)$, we have to calculate the derivatives in (\ref{Stationary Condition}b). This leads to
\begin{subequations}\label{Stationary Power}
\begin{align*}
    \frac{\partial\mathcal{L}}{\partial P_1(i)} =        &-\frac{1}{N\mathrm{ln}2} \Big[ \big \{(1-\mu_1)q_1(i)-t(i)(\mu_1-\mu_2)q_3(i) \big \}\nonumber  \\
&\times \frac{S_1(i)}{1\Add P_1(i)S_1(i)}\Add \big\{t(i)(\mu_1\Minus \mu_2)\Add 1\Minus \mu_1 \big\}q_3(i) \nonumber \\
&\times \frac{S_1(i)}{1+P_1(i)S_1(i)+P_2(i)S_2(i)} \Big ] \nonumber \\
&+ \gamma \frac{1}{N} (q_1(i)+q_3(i))  - \nu_1(i) =0 \quad \tag{\stepcounter{equation}\theequation}\\
    \frac{\partial\mathcal{L}}{\partial P_2(i)} \Equal        &-\frac{1}{N\mathrm{ln}2} \Big[ \big \{ (1\Minus\mu_2)q_2(i) \Add(1\Minus t(i))(\mu_1\Minus \mu_2)q_3(i) \big\} \nonumber \\
&\times\frac{S_2(i)}{1\Add P_2(i)S_2(i)} \Add \big\{t(i)(\mu_1\Minus\mu_2)\Add 1\Minus \mu_1 \big\}q_3(i) \nonumber \\
&\times\frac{S_2(i)}{1+P_1(i)S_1(i)+P_2(i)S_2(i)} \Big]\nonumber \\
&+\gamma \frac{1}{N} (q_2(i)+q_3(i))  - \nu_2(i)=0 \quad\tag{\stepcounter{equation}\theequation}\\
    \frac{\partial\mathcal{L}}{\partial P_r(i)} =        &-\frac{1}{N\mathrm{ln}2} \Big[ \mu_2\left(q_4(i)+q_6(i)\right)\frac{S_1(i)}{1+P_r(i)S_1(i)}\nonumber \\ &+\mu_1\left(q_5(i)+q_6(i)\right)\frac{S_2(i)}{1+P_r(i)S_2(i)} \Big]\nonumber \\
&+\gamma \frac{1}{N} (q_4(i)+q_5(i)+q_6(i))  - \nu_r(i)=0 \tag{\stepcounter{equation}\theequation}
\end{align*}
\end{subequations}
The above conditions allow the derivation of the optimal powers for each transmission mode in each time slot.
For instance, in order to determine the transmit power of user 1 in  transmission mode $\mathcal{M}_1$, we
assume $q_1^*(i)=1$. From constraint $\mathrm{C4}$ in (\ref{AdaptProbMin}), we obtained that the other
selection variables are zero and therefore $q_3^*(i)=0$. Moreover, if $\mathcal{M}_1$ is selected
then $P_1^*(i)\neq 0$ and thus from (\ref{Complementary Slackness}d), we obtain $\nu_1^*(i)=0$.
Substituting these results in (\ref{Stationary Power}a), we obtain
\begin{IEEEeqnarray}{lll}	\label{eq_11}	
		  P_1^{\mathcal{M}_1} (i) = \left[\frac{1-\mu_1}{\gamma \mathrm{ln}2}-\frac{1}{S_1(i)}\right]^+,
\end{IEEEeqnarray}
where $[x]^+=\max\{0,x\}$. In a similar manner, we obtain the optimal powers for user 2 in mode $\mathcal{M}_2$, and the
optimal powers of the relay in modes $\mathcal{M}_4$ and $\mathcal{M}_5$ as follows:
\begin{IEEEeqnarray}{lll}\label{P245}		
		  P_2^{\mathcal{M}_2} (i) = \left[\frac{1-\mu_2}{\gamma \mathrm{ln}2}-\frac{1}{S_2(i)}\right]^+ \IEEEyesnumber\IEEEyessubnumber\\
P_r^{\mathcal{M}_4} (i) = \left[\frac{\mu_2}{\gamma \mathrm{ln}2}-\frac{1}{S_1(i)}\right]^+ \IEEEyessubnumber\\
P_r^{\mathcal{M}_5} (i) = \left[\frac{\mu_1}{\gamma \mathrm{ln}2}-\frac{1}{S_2(i)}\right]^+ \IEEEyessubnumber
\end{IEEEeqnarray}
In order to obtain the optimal powers of user 1 and user 2 in mode $\mathcal{M}_3$, we assume $q_3^*(i)=1$.
From $\mathrm{C4}$ in (\ref{AdaptProbMin}), we obtain that the other selection variables   are zero, and
therefore $q_1^*(i)=0$ and $q_2^*(i)=0$. We  note that if one of the powers of user 2 and user 1 is zero
mode $\mathcal{M}_3$ is identical to modes $\mathcal{M}_1$ and $\mathcal{M}_2$, respectively, and for that
case the optimal powers are already given by (\ref{eq_11}) and (\ref{P245}a), respectively. For the case
when $P_1^*(i)\neq 0$ and $P_2^*(i)\neq 0$, we obtain $\nu_1^*(i)=0$ and $\nu_2^*(i)=0$  from (\ref{Complementary Slackness}d).
Furthermore, for $q_3^*(i)=1$, we will show in Appendix \ref{AppKKT}.C that   $t(i)$ can only take the boundary values, i.e.,
zero or one, and cannot be in between. Hence, if we assume $t(i)=0$, from (\ref{Stationary Power}a)
and (\ref{Stationary Power}b), we obtain
\begin{IEEEeqnarray}{lll}\label{PowerM3}
&-\frac{1-\mu_1}{\mathrm{ln}2} \frac{S_1(i)}{1+P_1(i)S_1(i)+P_2(i)S_2(i)} + \gamma  =0 \quad \IEEEyesnumber\IEEEyessubnumber \\
&-\frac{1}{\mathrm{ln}2} \Big[ (\mu_1-\mu_2) \frac{S_2(i)}{1+P_2(i)S_2(i)} \nonumber \\
&+ (1-\mu_1) \frac{S_2(i)}{1+P_1(i)S_1(i)+P_2(i)S_2(i)} \Big]+ \gamma =0 \quad \IEEEyessubnumber
\end{IEEEeqnarray}
By substituting (\ref{PowerM3}a) in (\ref{PowerM3}b), we obtain $P_2^{\mathcal{M}_3} (i)$ and then we can derive $P_1^{\mathcal{M}_3} (i)$ from (\ref{PowerM3}a). This leads to
\begin{IEEEeqnarray}{lll}  \label{PM3-t0}
P_1^{\mathcal{M}_3} (i) =
\begin{cases}
P_1^{\mathcal{M}_1} (i),  \qquad\qquad\quad \mathrm{if}\,\, S_2 \leq \frac{S_1}{\frac{\mu_1-\mu_2}{\gamma \mathrm{ln}2}S_1+1} \\
\left[\frac{1-\mu_1}{\gamma \mathrm{ln}2}-\frac{\mu_1-\mu_2}{\gamma \mathrm{ln}2}\frac{1}{\frac{S_1(i)}{S_2(i)}-1}\right]^+,   \,\, \mathrm{otherwise}
\end{cases}\IEEEyesnumber\IEEEyessubnumber \\
P_2^{\mathcal{M}_3} (i) =
\begin{cases}
P_2^{\mathcal{M}_2} (i),  \qquad\qquad\quad \mathrm{if}\,\, S_2 \geq \frac{1-\mu_1}{1-\mu_2}S_1\\
\left[\frac{\mu_1-\mu_2}{\gamma \mathrm{ln}2}\frac{1}{1-\frac{S_2(i)}{S_1(i)}} - \frac{1}{S_2(i)}\right]^+,  \,\, \mathrm{otherwise}
\IEEEyessubnumber
\end{cases}\end{IEEEeqnarray}
Similarly, if we assume $t(i)=1$, we obtain
\begin{IEEEeqnarray}{lll} \label{PM3-t1}
P_1^{\mathcal{M}_3} (i) =
\begin{cases}
P_1^{\mathcal{M}_1} (i),  \qquad\qquad\quad \mathrm{if}\,\, S_2 \leq \frac{1-\mu_1}{1-\mu_2}S_1\\
\left[\frac{\mu_1-\mu_2}{\gamma \mathrm{ln}2}\frac{1}{\frac{S_1(i)}{S_2(i)}-1} - \frac{1}{S_1(i)}\right]^+,  \,\, \mathrm{otherwise}
\end{cases}\IEEEyesnumber\IEEEyessubnumber \\
P_2^{\mathcal{M}_3} (i) =
\begin{cases}
P_2^{\mathcal{M}_2} (i),  \qquad\qquad\quad \mathrm{if}\,\, S_2 \geq \frac{S_1}{\frac{\mu_1-\mu_2}{\gamma \mathrm{ln}2}S_1+1} \\
\left[\frac{1-\mu_2}{\gamma \mathrm{ln}2}-\frac{\mu_1-\mu_2}{\gamma \mathrm{ln}2}\frac{1}{1-\frac{S_2(i)}{S_1(i)}}\right]^+,\,\, \mathrm{otherwise}
\end{cases} \IEEEyessubnumber
\end{IEEEeqnarray}
We note that when $P_1^{\mathcal{M}_3} (i)=P_1^{\mathcal{M}_1} (i)$, we obtain $P_2^{\mathcal{M}_3} (i)=0$ which
means that mode $\mathcal{M}_3$ is identical to mode $\mathcal{M}_1$. Thus, there is no difference between both
modes so we select $\mathcal{M}_1$. In Figs \ref{FigSRegion} a) and \ref{FigSRegion} b),  the comparison of  $\Lambda_1(i),\Lambda_2(i)$,
and $\Lambda_3(i)$ is illustrated in the space of $(S_1(i),S_2(i))$. Moreover, the shaded area represents the region in
which the powers of users 1 and 2 are zero for $\mathcal{M}_1,\mathcal{M}_2$, and $\mathcal{M}_3$.
\begin{figure*}[!t]
\normalsize
\centering
\psfrag{S1}[c][c][0.6]{$S_1(i)$}
\psfrag{S2}[c][c][0.6]{$S_2(i)$}
\psfrag{xa}[c][c][0.6]{$\frac{\gamma \mathrm{ln}2}{1-\mu_1}$}
\psfrag{ya}[c][c][0.6]{$\frac{\gamma \mathrm{ln}2}{1-\mu_2}$}
\psfrag{xb}[c][c][0.6]{$\frac{\gamma \mathrm{ln}2}{1-\mu_1}$}
\psfrag{yb}[c][c][0.6]{$\frac{\gamma \mathrm{ln}2}{1-\mu_2}$}
\psfrag{xc}[c][c][0.6]{$\frac{\gamma \mathrm{ln}2}{\mu_2}$}
\psfrag{yc}[c][c][0.6]{$\frac{\gamma \mathrm{ln}2}{\mu_1}$}
\psfrag{Aa1}[c][c][0.75]{$\Lambda_2(i)=\Lambda_3(i)>\Lambda_1(i)$}
\psfrag{Aa21}[c][b][0.7]{$\Lambda_3(i)>\Lambda_1(i)$}
\psfrag{Aa22}[c][c][0.7]{$\Lambda_3(i)>\Lambda_2(i)$}
\psfrag{Aa3}[c][c][0.75]{$\Lambda_1(i)=\Lambda_3(i)>\Lambda_2(i)$}
\psfrag{Bb1}[c][c][0.75]{$\Lambda_2(i)=\Lambda_3(i)$}
\psfrag{Bb12}[c][c][0.75]{$\qquad>\Lambda_1(i)$}
\psfrag{Bb21}[c][c][0.65]{$\Lambda_3(i)>\Lambda_1(i)$}
\psfrag{Bb22}[c][t][0.65]{$\Lambda_3(i)>\Lambda_2(i)$}
\psfrag{Bb3}[c][c][0.75]{$\Lambda_1(i)=\Lambda_3(i)>\Lambda_2(i)$}
\psfrag{Cc21}[c][c][0.75]{$\Lambda_6(i)>\Lambda_4(i)$}
\psfrag{Cc22}[c][c][0.75]{$\Lambda_6(i)>\Lambda_5(i)$}
\psfrag{a}[c][c][0.75]{a) $\mu_1>\mu_2$}
\psfrag{b}[c][c][0.75]{b) $\mu_1<\mu_2$}
\psfrag{c}[c][c][0.75]{c)}
\includegraphics[width=6.5 in]{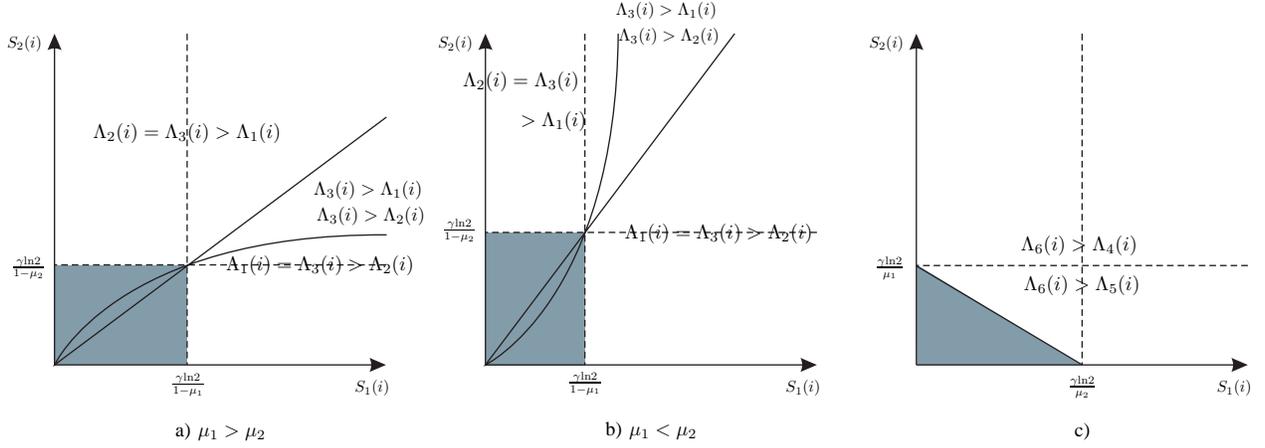}
\caption{Comparison of the selection metrics in the space of $(S_1,S_2)$: a) Comparison of $\Lambda_1(i),\Lambda_2(i)$, and $\Lambda_3(i)$ if $\mu_1>\mu_2$, b) comparison of $\Lambda_1(i),\Lambda_2(i)$, and $\Lambda_3(i)$ if $\mu_1<\mu_2$, c) comparison of  $\Lambda_4(i),\Lambda_5(i)$, and $\Lambda_6(i)$.}
\label{FigSRegion}
\vspace*{4pt}
\end{figure*}

For mode $\mathcal{M}_6$, we assume $q_6^*(i)=1$. From constraint $\mathrm{C4}$ in (\ref{AdaptProbMin}), we obtain
that the other selection variables are zero and therefore $q_4^*(i)=0$ and $q_5^*(i)=0$. Moreover,
if $q_6^*(i)=1$ then $P_r^*(i)\neq 0$ and thus from (\ref{Complementary Slackness}d), we obtain $\nu_r^*(i)=0$.
Using these results in (\ref{Stationary Power}c), we obtain
\begin{IEEEeqnarray}{lll}\label{PowerM6}
\mu_2\frac{S_1(i)}{1+P_r(i)S_1(i)} + \mu_1\frac{S_2(i)}{1+P_r(i)S_2(i)} = \gamma \mathrm{ln}2
\end{IEEEeqnarray}
The above equation is a quadratic equation and has two solutions for $P_r(i)$. However, since we have $P_r(i)\geq 0$, we can conclude that the left hand side of (\ref{PowerM6}) is monotonically decreasing in $P_r(i)$. Thus, if $\mu_2S_1(i)+\mu_1S_2(i)>\gamma \mathrm{ln}2$, we have a unique positive solution for $P_r(i)$ which is the maximum of the two roots of (\ref{PowerM6}). Thus, we obtain
\begin{IEEEeqnarray}{lll}\label{PM6}
P_r^{\mathcal{M}_6} (i) = \left [ \frac{-b+\sqrt{b^2-4ac}}{2a} \right]^+, 
\end{IEEEeqnarray}
where $a=\gamma \mathrm{ln}2 S_1(i)S_2(i), b= \gamma \mathrm{ln}2 (S_1(i)+S_2(i)) - (\mu_1+\mu_2)S_1(i)S_2(i)$, and $c=\gamma \mathrm{ln}2 - \mu_1 S_2(i) - \mu_2 S_1(i)$.

In  Fig. \ref{FigSRegion} c), the comparison between selection metrics $\Lambda_4(i),\Lambda_5(i),$ and $\Lambda_6(i)$ is
illustrated in the space of $(S_1(i),S_2(i))$. We note that
$\Lambda_6(i)\geq \Lambda_4(i)$ and $\Lambda_6(i)\geq \Lambda_5(i)$ hold and the inequalities hold with equality if $S_2(i)=0$ and $S_1(i)=0$, respectively, which happen with zero probability for time-continuous fading.
To prove $\Lambda_6(i)\geq \Lambda_4(i)$, from (\ref{MET}), we obtain
\begin{IEEEeqnarray}{rCl}
   \Lambda_6(i)  &=& \mu_1C_{r2}(i) + \mu_2C_{r1}(i) -\gamma P_r(i) \big |_{P_r(i)=P_r^{\mathcal{M}_6}(i)} \nonumber \\
 &\overset{(a)} {\geq}& \mu_1C_{r2}(i) + \mu_2C_{r1}(i) -\gamma P_r(i)\big |_{P_r(i)=P_r^{\mathcal{M}_4}(i)} \nonumber \\
 &\overset{(b)} {\geq}& \mu_2C_{r1}(i) -\gamma P_r(i) \big |_{P_r(i)=P_r^{\mathcal{M}_4}(i)} = \Lambda_4(i),
\end{IEEEeqnarray}
where   $(a)$ follows from the fact that $P_r^{\mathcal{M}_6}(i)$ maximizes $\Lambda_6(i)$ and   $(b)$ follows from $\mu_1C_{r2}(i)\geq 0$. The two inequalities $(a)$ and $(b)$ hold with equality only if $S_2(i)=0$ which happens with zero probability in time-continuous fading or if $\mu_1=0$. However,  in Appendix \ref{AppMURegion}, $\mu_1=0$  is shown to lead to a contradiction. Therefore, the optimal policy does not select $\mathcal{M}_4$ and $\mathcal{M}_5$ and selects only modes $\mathcal{M}_1,\mathcal{M}_2,\mathcal{M}_3$, and $\mathcal{M}_6$.

\subsection{Optimal $t^*(i)$}
To find the optimal $t(i)$, we assume $q^*_3(i)=1$ and calculate the stationary condition in (\ref{Stationary Condition}c). This leads to
\begin{IEEEeqnarray}{lll} \label{Stationary t}
    \frac{\partial\mathcal{L}}{\partial t(i)} =&-\frac{1}{N}(\mu_1-\mu_2) \left[ C_r(i) - C_{1r}(i) - C_{2r}(i) \right] \nonumber \\
&+\phi_1(i)-\phi_0(i)=0
\end{IEEEeqnarray}
Now, we investigate the following possible cases for $t^*(i)$:

\noindent
\textbf{Case 1:} If $0<t^*(i)<1$ then from (\ref{Complementary Slackness}e) and (\ref{Complementary Slackness}f), we have $\phi_l(i)=0, \,\, l=0,1$. Therefore, from (\ref{Stationary t}) and $C_r(i) - C_{1r}(i) - C_{2r}(i)\leq 0$, we obtain $\mu_1=\mu_2$. Then, from (\ref{Stationary Power}a) and (\ref{Stationary Power}b), we obtain
\begin{IEEEeqnarray}{rCl}\label{ContradicT}
    {\begin{cases}
     -\frac{1}{\mathrm{ln} 2} (1-\mu_1) \frac{S_1(i)}{1+P_1(i)S_1(i)+P_2(i)S_2(i)} +\gamma =0 \\
		 -\frac{1}{\mathrm{ln} 2} (1-\mu_1) \frac{S_2(i)}{1+P_1(i)S_1(i)+P_2(i)S_2(i)} +\gamma =0
    \end{cases}}
\end{IEEEeqnarray}
In Appendix \ref{AppMURegion}, we show that $\mu_1 \neq 1$, therefore, the above conditions can be satisfied simultaneously only if $S_1(i)=S_2(i)$, which, considering the randomness of the time-continuous channel gains, occurs with zero probability. Hence, the optimal $t(i)$ takes the boundary values, i.e., zero or one, and not values in between.

\noindent
\textbf{Case 2:} If $t^*(i)=0$, then from (\ref{Complementary Slackness}e), we obtain $\phi_1(i)=0$ and from (\ref{Dual Feasibility Condition}e), we obtain $\phi_0(i)\geq 0$. Combining these results into (\ref{Stationary t}), the necessary condition for  $t(i)=0$ is obtained as $\mu_1\geq \mu_2$.

\noindent
\textbf{Case 3:} If $t^*(i)=1$, then from (\ref{Complementary Slackness}f), we obtain $\phi_0(i)=0$ and from  (\ref{Dual Feasibility Condition}e), we obtain $\phi_1(i)\geq 0$. Combining these results into (\ref{Stationary t}), the necessary condition for $t(i)=1$ is obtained as $\mu_1 \leq  \mu_2$.

We note that if $\mu_1=\mu_2$, we obtain either
$P_1^{\mathcal{M}_3}(i)=0$ or $P_2^{\mathcal{M}_3}(i)=0,\,\,\forall i$. Therefore, mode
$\mathcal{M}_3$ is not selected and the value of $t(i)$ does not affect the sum rate.
Moreover, from the selection metrics in (\ref{MET}), we can conclude that $\mu_1>\mu_2$ and $\mu_1<\mu_2$
correspond to $\Omega_1>\Omega_2$ and $\Omega_1<\Omega_2$, respectively. Therefore, the optimal value of $t(i)$ is given by
\begin{IEEEeqnarray}{lll}
   t^*(i) = {\begin{cases} 0, & \Omega_1 \geq \Omega_2 \\
    1, & \Omega_1 < \Omega_2
\end{cases}}
\end{IEEEeqnarray}

Now, the optimal values of $q_k(i),P_j(i)$, and $t(i),\,\,\forall i,j,k$ are derived based on which Theorem \ref{AdaptProt} can be constructed. This completes the proof.


\section{Proof of Optimality of Binary Relaxation}
\label{AppBinRelax}

In this appendix, we prove that the optimal solution of the problem with the relaxed constraint,
$0\leq  q_k(i)\leq 1$, selects the boundary values of $q_k(i)$, i.e., zero or one.
Therefore, the binary relaxation does not change the solution of the problem. If one of the $q_k(i),\,\,k=1,\dots,6$,
adopts a non-binary value in the optimal solution, then in order to satisfy constraint $\mathrm{C4}$ in (\ref{AdaptProb}), there has to be at least one other non-binary selection variable in that time slot.
Assuming that the mode indices of the non-binary selection variables are $k'$ and $k''$ in the $i$-th time slot,
we obtain $\alpha_k(i)=0,\,\,k = 1,\dots,6$ from (\ref{Complementary Slackness}a), and $ \beta_{k'}(i)=0$ and
$ \beta_{k''}(i)=0$  from (\ref{Complementary Slackness}b). Then, by substituting these values into (\ref{Stationary Mode}),
we obtain
\begin{IEEEeqnarray}{lll}\label{BinRelax}
    \lambda(i) = \Lambda_{k'}(i)  \IEEEyesnumber\IEEEyessubnumber  \\
    \lambda(i)= \Lambda_{k''}(i)\IEEEyessubnumber  \\
   \lambda(i)-\beta_k(i) =  \Lambda_k(i), \quad k\neq k', k''. \IEEEyessubnumber
\end{IEEEeqnarray}
From (\ref{BinRelax}a) and  (\ref{BinRelax}b), we obtain $\Lambda_{k'}(i)=\Lambda_{k''}(i)$ and by subtracting (\ref{BinRelax}a) and  (\ref{BinRelax}b) from (\ref{BinRelax}c), we obtain
\begin{IEEEeqnarray}{rCl}
    \Lambda_{k'}(i) - \Lambda_k(i) &=& \beta_k(i), \quad \quad k\neq k', k'' \IEEEyesnumber \IEEEyessubnumber  \\
\Lambda_{k''}(i) - \Lambda_k(i) &=& \beta_k(i), \quad \quad k\neq k', k''. \IEEEyessubnumber
\end{IEEEeqnarray}
From the dual feasibility condition given in (\ref{Dual Feasibility Condition}b), we have $\beta_k(i)\geq 0$ which leads to $\Lambda_{k'}(i)=\Lambda_{k''}(i)\geq \Lambda_k(i)$.
However, as a result of the randomness of the time-continuous channel gains,  $\Pr\{\Lambda_{k'}(i)=\Lambda_{k''}(i)\} > 0$
holds for some transmission modes $\mathcal{M}_{k'}$ and $\mathcal{M}_{k''}$, if and only if we obtain
$\mu_1=0,1$ or $\mu_2=0,1$ which leads to a contradiction as shown in Appendix \ref{AppMURegion}. This completes the proof.


\section{Threshold Regions}
\label{AppMURegion}

In  this  appendix,  we  find  the  intervals  which  contain
the  optimal  value  of  $\mu_1$ and $\mu_2$.  We  note  that  for  different values  of  $\mu_1$ and $\mu_2$,  some  of  the  optimal  powers  derived in (\ref{eq_11}), (\ref{P245}), (\ref{PM3-t0}), (\ref{PM3-t1}), and (\ref{PM6})    are  zero  for  all  channel realizations.  For  example,  if   $\mu_1=1$,  we  obtain  $P^{\mathcal{M}_1}_1(i)=0,\,\,\forall i$
from (\ref{eq_11}). Fig. \ref{FigMURegion} illustrates the set of modes that can take positive powers with non-zero probability in the space of ($\mu_1,\mu_2$). In the following, we show that any values of $\mu_1$ and $\mu_2$ except $0<\mu_1<1$ and $0<\mu_2<1$
cannot lead to the optimal sum rate or violate constraints $\mathrm{C1}$ or $\mathrm{C2}$ in (\ref{AdaptProbMin}).

\noindent
\textbf{Case 1:} Sets  $\{\mathcal{M}_1,\mathcal{M}_2,\mathcal{M}_3\}$ and $\{\mathcal{M}_4,\mathcal{M}_5,\mathcal{M}_6\}$ lead to selection of either the transmission from the users to the relay or the transmission from the relay to the users, respectively,  for all time slots. This leads to violation of constraints $\mathrm{C1}$ and $\mathrm{C2}$ in (\ref{AdaptProbMin}) and thus the optimal values of $\mu_1$ and $\mu_2$ are not in this region.

\noindent
\textbf{Case 2:} In set  $\{\mathcal{M}_1,\mathcal{M}_4,\mathcal{M}_6\}$, both modes $\mathcal{M}_4$ and $\mathcal{M}_6$ need the transmission from user 2 to the relay which can not be realized in this set. Thus, this set leads to violation of constraint $\mathrm{C2}$ in (\ref{AdaptProbMin}). Similarly, in set $\{\mathcal{M}_2,\mathcal{M}_5,\mathcal{M}_6\}$,  both modes $\mathcal{M}_5$ and $\mathcal{M}_6$ require the transmission from user 1 to the relay which can not be selected in this set. Thus, this region of $\mu_1$ and $\mu_2$ leads to violation of constraint $\mathrm{C1}$ in (\ref{AdaptProbMin}).

\noindent
\textbf{Case 3:} In set $\{\mathcal{M}_1,\mathcal{M}_4,\mathcal{M}_5,\mathcal{M}_6\}$, there is no transmission from user 2 to
the relay. Therefore, the optimal values of $\mu_1$ and $\mu_2$ have to guarantee that modes $\mathcal{M}_4$
and $\mathcal{M}_6$ are not selected for any channel realization. However, from (\ref{MET}), we obtain
\begin{IEEEeqnarray}{rCl}
   \Lambda_6(i)  &=& \mu_1C_{r2}(i) + \mu_2C_{r1}(i) -\gamma P_r(i) \big |_{P_r(i)=P_r^{\mathcal{M}_6}(i)} \nonumber \\
 &\overset{(a)} {\geq}& \mu_1C_{r2}(i) + \mu_2C_{r1}(i) -\gamma P_r(i)\big |_{P_r(i)=P_r^{\mathcal{M}_5}(i)} \nonumber \\
 &\overset{(b)} {\geq}& \mu_1C_{r2}(i) -\gamma P_r(i) \big |_{P_r(i)=P_r^{\mathcal{M}_5}(i)} = \Lambda_5(i),
\end{IEEEeqnarray}
where $(a)$ follows from the fact that $P_r^{\mathcal{M}_6}(i)$ maximizes $\Lambda_6(i)$ and   $(b)$ follows
from $\mu_2C_{r1}(i)\geq 0$. The two inequalities $(a)$ and $(b)$ hold with equality only if $S_1(i)=0$ which happens with zero probability for time-continuous fading, or $\mu_2=0$ which is not included in this region. Therefore, mode $\mathcal{M}_6$ is selected in this region which leads to violation of constraint $\mathrm{C2}$ in (\ref{AdaptProbMin}). A similar statement is true for set $\{\mathcal{M}_2,\mathcal{M}_4,\mathcal{M}_5,\mathcal{M}_6\}$. Thus, the optimal values of $\mu_1$ and $\mu_2$ cannot be in these two regions.

\noindent
\textbf{Case 4:} In set $\{\mathcal{M}_1,\mathcal{M}_2,\mathcal{M}_3,\mathcal{M}_4,\mathcal{M}_6\}$,
we obtain
\begin{IEEEeqnarray}{rCl}
   \Lambda_6(i) &=& \mu_1C_{r2}(i) + \mu_2 C_{r1}(i)-\gamma P_r(i)\big |_{P_r(i)=P_r^{\mathcal{M}_6}(i)} \nonumber \\
 &\overset{(a)} {\leq}& \mu_2 C_{r1}(i)-\gamma P_r(i) \big |_{P_r(i)=P_r^{\mathcal{M}_6}(i)} \nonumber \\
 &\overset{(b)} {\leq}& \mu_2 C_{r1}(i)-\gamma P_r(i) \big |_{P_r(i)=P_r^{\mathcal{M}_4}(i)} = \Lambda_4(i)
\end{IEEEeqnarray}
where inequality $(a)$ comes from the fact that $\mu_1C_{r2}(i)\leq 0$ and the equality holds when $S_2(i)=0$ which happens
with zero probability, or $\mu_1=0$. Inequality $(b)$ holds since $P_r^{\mathcal{M}_4}(i)$ maximizes $\Lambda_4(i)$ and
holds with equality only if $P_r^{\mathcal{M}_4}(i)=P_r^{\mathcal{M}_6}(i)$ and consequently $\mu_1=0$. If $\mu_1\neq0$, mode $\mathcal{M}_6$ is not selected and there is no transmission from the relay to user 2. Therefore, the optimal values of $\mu_1$ and $\mu_2$ have to guarantee that modes $\mathcal{M}_1$ and $\mathcal{M}_3$ are not selected for any channel realization. Thus, we obtain $\mu_1=1$ which is not contained in this region. If $\mu_1=0$,
from (\ref{MET}), we obtain
\begin{IEEEeqnarray}{rCl}
   \Lambda_6(i) = \Lambda_4(i) &=& \mu_2C_{r1}(i) -\gamma P_r(i)\big |_{P_r(i)=P_r^{\mathcal{M}_4}(i)} \nonumber \\
 &\overset{(a)} {\leq}& C_{r1}(i) -\gamma P_r(i)\big |_{P_r(i)=P_r^{\mathcal{M}_4}(i)} \nonumber \\
 &\overset{(b)} {\leq}& C_{1r}(i)-\gamma P_1(i) \big |_{P_1(i)=P_1^{\mathcal{M}_1}(i)} = \Lambda_1(i) \quad\,\,
\end{IEEEeqnarray}
where both inequalities $(a)$ and $(b)$ hold with equality only if $\mu_2=1$. If $\mu_2\neq 1$, modes $\mathcal{M}_4$ and $\mathcal{M}_6$ are not selected. Thus, there is no transmission from the relay to the users which leads to violation of $\mathrm{C1}$ and $\mathrm{C2}$ in (\ref{AdaptProbMin}). If $\mu_2=1$, we obtain $P_2^{\mathcal{M}_2}(i)=0$, thus mode $\mathcal{M}_2$ cannot be selected and either $P_1^{\mathcal{M}_3}(i)=0$ or $P_2^{\mathcal{M}_3}(i)=0$, thus mode $\mathcal{M}_3$ cannot be selected either. Since both modes $\mathcal{M}_4$ and $\mathcal{M}_6$ require the transmission from user 2 to the relay, and both modes $\mathcal{M}_2$ and $\mathcal{M}_3$ are not selected, constraint $\mathrm{C2}$ in (\ref{AdaptProbMin}) is violated and $\mu_1=0$ and $\mu_2=1$ cannot be optimal. A similar statement is true for set $\{\mathcal{M}_1,\mathcal{M}_2,\mathcal{M}_3,\mathcal{M}_5,\mathcal{M}_6\}$. Therefore, the optimal values of $\mu_1$ and $\mu_2$ are not in this region.

Hence, set $\{\mathcal{M}_1,\mathcal{M}_2,\mathcal{M}_3,\mathcal{M}_4,\mathcal{M}_5,\mathcal{M}_6\}$ contains the optimal
values of $\mu_1$ and $\mu_2$, i.e., $0<\mu_1<1$ and $0<\mu_2<1$. This completes the proof.
\begin{figure}
\centering
\psfrag{M1}[c][c][0.75]{$\mu_1$}
\psfrag{M2}[c][c][0.75]{$\mu_2$}
\psfrag{Mu1}[c][c][0.75]{$\mu_1 = 1$}
\psfrag{M21}[c][c][0.75]{$\,\,\mu_2=1$}
\psfrag{S1}[c][c][0.75]{$\big\{\mathcal{M}_1,\mathcal{M}_4,\mathcal{M}_6\big\}$}
\psfrag{S2}[c][c][0.75]{\begin{tabular}{c}
   $\big\{\mathcal{M}_1,\mathcal{M}_2,\mathcal{M}_3,$\\
  $\quad\qquad\mathcal{M}_4,\mathcal{M}_6\big\}$
\end{tabular}}
\psfrag{S3}[c][c][0.75]{$\big\{\mathcal{M}_1,\mathcal{M}_2,\mathcal{M}_3\big\}$}
\psfrag{S4}[c][c][0.75]{\begin{tabular}{c}
  $\big\{\mathcal{M}_1,\mathcal{M}_2,$\\
   $\,\,\,\mathcal{M}_3,\mathcal{M}_5,$ \\
$\quad\qquad\mathcal{M}_6\big\}$
\end{tabular}}
\psfrag{S5}[c][c][0.75]{\begin{tabular}{c}
   $\big\{\mathcal{M}_1,\mathcal{M}_2,$\\
   $\,\,\,\mathcal{M}_3,\mathcal{M}_4,$ \\
$\,\,\,\mathcal{M}_5,\mathcal{M}_6\big\}$
\end{tabular}}
\psfrag{S6}[c][c][0.75]{\begin{tabular}{c}
   $\big\{\mathcal{M}_1,\mathcal{M}_4,$\\
$\,\,\,\mathcal{M}_5,\mathcal{M}_6\big\}$
\end{tabular}}
\psfrag{S7}[c][c][0.75]{$\big\{\mathcal{M}_4,\mathcal{M}_5,\mathcal{M}_6\big\}$}
\psfrag{S8}[c][c][0.75]{\begin{tabular}{c}
   $\big\{\mathcal{M}_2,\mathcal{M}_4,$\\
$\,\,\,\mathcal{M}_5,\mathcal{M}_6\big\}$
\end{tabular}}
\psfrag{S9}[c][c][0.75]{$\big\{\mathcal{M}_2,\mathcal{M}_5,\mathcal{M}_6\big\}$}
\includegraphics[width=3 in]{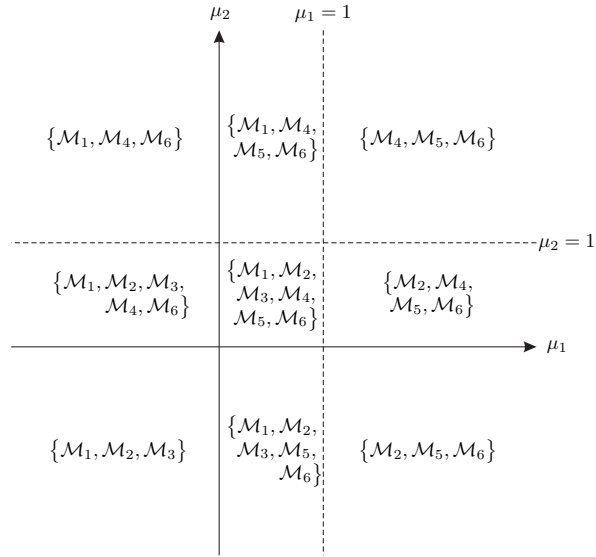}
\caption{Modes with non-negative powers in the space of $(\mu_1,\mu_2)$.}
\label{FigMURegion}
\end{figure}



\end{document}